\documentclass{article}
\usepackage{fullpage}
\usepackage{times}
\usepackage{color}
\usepackage{amsmath,amssymb,amsthm, amsfonts, enumitem,centernot,mathtools,extarrows,mathdots,appendix}
\usepackage{comment,MnSymbol}
\usepackage[ruled,vlined]{algorithm2e}
\usepackage{physics}
\usepackage{tikz-cd}
\usepackage{mdframed}

\newtheorem{theorem}{Theorem}

\newtheorem{proposition}{Proposition}
\newtheorem{corollary}{Corollary}
\newtheorem{eg}{Example}

\begin{document}
\title{Cayley's First Hyperdeterminant is an Entanglement Measure}
\author{Isaac Dobes and Naihuan Jing}
\maketitle

\begin{abstract}
    Previously, it was shown that both the concurrence and $n$-tangle on $2n$-qubit pure quantum states can be expressed in terms of Cayley's first hyperdeterminant \cite{dobes2024qubits}, indicating that Cayley's first hyperdeterminant, denoted $\mathrm{hdet}$, captures some aspects of a state's $2n$-way entanglement. In this paper, we rigorously prove that on both pure and mixed states, $|\mathrm{hdet}|^{2/d}$ is identically zero on separable states, is an LU invariant, and is non-increasing on average under LOCC, thus demonstrating that $|\mathrm{hdet}|^{d/2}$ is a physically meaningful and legitimate entanglement measure. Moreover, we discuss a few key examples to illustrate the particular type of entanglement Cayley's first hyperdeterminant is detecting: genuine full $d$-level GHZ-type entanglement across all $2n$ parties. Combined, this establishes Cayley's first hyperdeterminant (or $|\mathrm{hdet}|^{2/d}$ to be precise), as a physically significant generalization of the $G$-concurrence and the $n$-tangle to $2n$-qudit states. 
\end{abstract}

\section{Introduction}
To compare classes of entangled quantum states, as well as states within the same entanglement class, it would be ideal to derive a function that continuously quantifies all possible levels of the entanglement for any arbitrary state. For pure states in bipartite quantum systems, it is widely agreed that their precise level of entanglement is derived from their Schmidt coefficients \cite{nielsen2010quantum,ma2024multipartite}. However, for states of multipartite quantum states, the situation is more complicated and ambiguous. 

Indeed, at the time of writing, there is currently no \textit{universally} agreed upon mathematical structure which completely determines an arbitrary multipartite state's entanglement, or a complete set of criteria that unambiguously determines whether or not a function defined on arbitrary multipartite quantum states is a valid entanglement measure. Furthermore, this current status is unlikely to change \cite{walter2016multipartite} due to the fact that multipartite quantum states can be entangled in different, inequivalent ways \cite{dur2000three}. Nonetheless, there are a few key properties that most researcher's agree any legitimate measure of entanglement ought to satisfy. 

In the late 1990's numerous landmark papers were published on quantifying entanglement, and in these papers a variety of physically significant properties were listed as axioms that any proposed entanglement measure $E$ ought to satisfy; e.g, \cite{bennett1996mixed,vedral1997quantifying,vidal2000entanglement,bennett2000exact} (see also \cite{eltschka2014quantifying,bengtsson2017geometry} for more recent summaries and analysis of aforementioned papers). While the properties listed varies depending on the paper, three properties that show up in nearly all such papers are the following:
\begin{enumerate}
    \item $E$ is nonnegative and $0$ on separable states
    \item $E$ is invariant under local unitary (LU) operations
    \item $E$ is non-increasing on average under local operations and classical communications (LOCC)
\end{enumerate}
Property $1.$ is due to the fact that by definition separable states posses no entanglement, hence any entanglement measure $E$ should be minimized on such states. Property $2.$ is salient because LU operations are reversible local change of reference frames independently applied to each party, hence no information is transferred between parties and their correlations are unaffected, in which case $E$ should remain unchanged \cite{monras2011entanglement}. Property $3.$ is also crucial because while local measurements can in some instances increase the amount of entanglement in a partially entangled state \cite{bennett1996concentrating}, entanglement is a non-local resource, so its expected value over all measurement outcomes cannot be increased by local operations and classical communication alone. Note also that in Property $3.$, when we say "on average," we precisely mean that if $\rho$ (which may represent a pure or mixed state) is transformed via a LOCC protocol into the mixture of states $\sum\limits_jp_j\rho_j$, then 
\[\sum\limits_jp_jE(\rho_j)\leq E(\rho).\]
If $E$ satisfies Properties $2.$ and $3.$, then $E$ is referred to as an "entanglement monotone" \cite{vidal2000entanglement}, and if $E$ satisfies all three properties above, then it is reasonable to refer to it as an "entanglement measure"; for the sake of convenience, we refer to the above three properties as the "entanglement measure axioms". 

Note that the entanglement measure axioms are by no means the only properties one might desire an entanglement measure to satisfy. For instance, convexity and additivity (or at least sub-additivity) are other properties typically desired in an entanglement measure, however it is the above three that are of the most fundamental importance and found throughout virtually all of the relevant literature \cite{eltschka2014quantifying,zhahir2022entanglement}. Note also that Property $1.$ is sometimes strengthened to: $E(\rho)=0$ if and only if $\rho$ is separable, however this is a very strong condition and there are numerous well-established entanglement measures which do not satisfy this, e.g. the negativity, the entanglement of purification (which does not even satisfy Property $1.$ \cite{bengtsson2017geometry}), and the concurrence/$n$-tangle on $2n$-qubit pure quantum states. Property $3.$ is also often modified to: $E$ is non-increasing on average under stochastic local operations and classical communications (SLOCC) \cite{eltschka2014quantifying,bengtsson2017geometry}, whereby SLOCC we mean that $\rho$ can be transformed into the state $\sigma$ via a LOCC protocol with nonzero probability. 

Two popular entanglement measures are the concurrence and the $n$-tangle. The concurrence was originally formulated to define the entanglement of formation \cite{hill1997entanglement,wootters1998entanglement}, and on pure states of $2n$-qubits it is given by 
\[C(\ket{\psi}) := |\braket{\widetilde{\psi}}{\psi}|,\]
where 
\[\ket{\widetilde{\psi}} := \sigma_y^{\otimes 2n}\ket{\psi^*}\]
with $\sigma_y$ denoting the Pauli matrix $\left[\begin{array}{cc}
    0 & -i \\
    i & 0
\end{array}\right]$. The $3$-tangle was originally formulated to detect GHZ-type entanglement of $3$-qubit states \cite{dur2000three}, and then soon afterwards it was extended to pure states of $2n$-qubits \cite{wong2001potential}. In addition to demonstrating the detection of $2n$-qubit GHZ-type entanglement in pure states of $2n$-qubits, it was also shown in \cite{wong2001potential} that in fact the $n$-tangle is the square of the concurrence on pure states of $2n$-qubits; that is, 
\[\tau_n(\ket{\psi}) = C(\ket{\psi})^2 = |\braket{\widetilde{\psi}}{\psi}|^2.\] 

About 20 years later, it was shown in \cite{dobes2024qubits} that the concurrence, and hence the $n$-tangle, may be expressed in terms of Cayley's first hyperdeterminant when viewed as a quadratic form in the coefficients of a pure $2n$-qubit state $\ket{\psi}$. In particular, I. Dobes and N. Jing prove using elementary combinatorics that 
\begin{equation}\label{Ent - Pauli relationship}
    \widehat{\mathrm{Ent}}_n = \frac{(-1)^n}{2}\sigma_y^{\otimes 2n},
\end{equation}
with $\widehat{\mathrm{Ent}}_n$ denoting the matrix of the quadratic form given by Cayley's first hyperdeterminant. Denoting Cayley's first hyperdeterminant as $\mathrm{hdet}$, from equation \eqref{Ent - Pauli relationship} it is then easy to show that 
\[C(\ket{\psi}) = 2|\mathrm{hdet}(\widehat{\psi})|\]
and hence
\[\tau_n(\ket{\psi}) = 4|\mathrm{hdet}(\widehat{\psi})|^2,\]
where $\widehat{\psi}$ denotes the tensor representation of $\ket{\psi}$ (further discussed in the next section). These results, in addition to the known fact that the hyperdeterminant is invariant under (special) unitary operations \cite{jing2014slocc}, indicate that Cayley's first hyperdeterminant captures at least \textit{some} aspects of a quantum state's entanglement. Unlike the concurrence and the $n$-tangle, however, Cayley's first hyperdeterminant is also defined and non-trivial on pure state of $2n$-qudits for any $d\geq 2$. This suggests that Cayley's first hyperdeterminant may in fact be a generalization of the concurrence and the $n$-tangle. Indeed, in this paper we will rigorously prove that $|\mathrm{hdet}|^{2/d}$ satisfies the entanglement measure axioms on states of $2n$-qudits, and is thus a legitimate and physically meaningful entanglement measure. Moreover, we will illustrate through example that the physical interpretation of Cayley's first hyperdeterminant as an entanglement measure is analogous both to a higher-dimensional $n$-tangle, as well as a multipartite $G$-concurrence (an extension of the concurrence to higher dimensional bipartite quantum states, first proposed in \cite{gour2005mixed}). 

\section{Preliminaries}
\subsection{Hypermatrix Algebra}
Let $F$ be a field and suppose $V$ and $W$ are finite-dimensional vector spaces of dimensions $m$ and $n$, respectively. Just as a linear transformation $T:V\rightarrow W$ in a fixed basis uniquely corresponds to a matrix $A\in F^{m\times n}$, similarly a tensor $T\in V_1\otimes...\otimes V_d$ in a fixed basis with $\dim(V_i)=n_i$ for each $1\leq i\leq d$, uniquely corresponds to a hypermatrix $A\in F^{n_1\times...\times n_d}$. The positive integer $d$ is called the order of a hypermatrix. The Frobenius norm of a hypermatrix $A = [a_{i_1...i_d}]_{n_1\times...\times n_d}\in F^{n_1\times...\times n_d}$ is defined as 
\[\|A\|_F := \left(\sum\limits_{i_1,...,i_n=1}^{n_1,...,n_d}|a_{i_1...i_d}|^2\right)^{1/2}.\]
Just as the transpose permutes the subscripts of a matrix, the $\pi$-transpose permutes the subscripts of a hypermatrix; that is, if $A = [a_{i_1...i_d}]_{n_1\times...\times n_d}\in F^{n_1\times...\times n_d}$, then the $\pi$-transpose of $A$ is defined as the hypermatrix 
\[A^{\pi} := [a_{i_{\pi(1)}...i_{\pi(d)}}]_{n_{\pi(1)}\times...\times n_{\pi(d)}}\in F^{n_{\pi(1)}\times...\times n_{\pi(d)}}\]
where $\pi\in S_d$. 

Suppose $A = [a_{i_1...i_d}]_{n_1\times...\times n_d}\in F^{n_1\times...\times n_d}$ and $B = [b_{j_1...j_e}]_{m_1\times...\times n_e}\in F^{m_1\times...\times m_e}$ are hypermatrices of order $d$ and $e$. The outer product of $A$ and $B$, denoted $A\circ B$, is the hypermatrix of order $d+e$ whose $(i_1,...,i_d,j_1,...,j_e)$ coordinate is given by $a_{i_1,...,i_d}b_{j_1,...,j_e}$. For a hypermatrix $A\in F^{n_1\times n_2\times...\times n_d}$ and matrices $X_1\in F^{m_1\times n_1},...,X_d\in F^{m_d\times n_d}$, the multilinear matrix multiplication of $(X_1,...,X_d)$ with $A$ is defined to be the order $m_1\times...\times m_d$ hypermatrix $(X_1,...,X_d)*A := A'$ such that
\[A'_{i_1i_2...i_d} = \sum\limits_{j_1,j_2,...,j_d=1}^{n_1,n_2,...,n_d}(X_1)_{i_1j_1}...(X_d)_{i_dj_d}A_{j_1j_2...j_d}.\]
Multilinear matrix multiplication is linear in terms of the matrices in both parts; that is, if $\alpha,\beta\in F$, $X_1,Y_1\in F^{m_1\times n_1}$;...; $X_d,Y_d\in F^{m_d\times n_d}$; and $A,B\in F^{n_1\times n_2\times...\times n_d}$; then 
\[(X_1,...,X_d)*(\alpha A + \beta B) = \alpha(X_1,...,X_d)*A + \beta(Y_1,...,Y_d)*B\]
and 
\[[\alpha(X_1,...,X_d)+\beta(Y_1...,Y_d)]*A = \alpha(X_1,...,X_d)*A + \beta(Y_1,...,Y_d)*A.\]
Next, let $k:[d]\rightarrow \mathbb{N}$. The outer product interacts with multilinear matrix multiplication in the following way
\begin{eqnarray}\label{outerproduct - multilinmatmult - relation}
    \begin{aligned}
        &(X_{1_1},...,X_{1_{k(1)}},...,X_{d_1},...,X_{d_{k(d)}})*(A_1\circ...\circ A_d) \\
        &\qquad = (X_{1_1},...,X_{1_{k(1)}})*A_1\circ...\circ (X_{d_1},...,X_{d_{k(d)}})*A_d
    \end{aligned} 
\end{eqnarray}
where $X_{i_j}\in F^{m^{(i)}_j\times n^{(i)}_j}$ and $A_i\in F^{n^{(i)}_1\times....\times n^{(i)}_{k(i)}}$. Consequently, noting that any hypermatrix can be expressed as a sum of the outer product of rank-$1$ tensors \cite{lim2013tensors}, i.e. $A\in F^{n_1\times...\times n_d}$ is of the form
\[A = \sum\limits_{k=1}^r\alpha_k(v_1^{(k)}\circ...\circ v_d^{(k)})\]
with $v_i^{(k)}\in F^{n_i}$ for each $i$ and $k$, then equation \eqref{outerproduct - multilinmatmult - relation} and the linearity of multilinear matrix multiplication implies that 
\begin{equation}\label{outerproduct - multilinmatmult - special case}
    (X_1,...,X_d)*A = \sum\limits_{k=1}^r\alpha_k(X_1v_1^{(k)})\circ...\circ (X_dv_d^{(k)}).
\end{equation}
for any matrices $X_1\in F^{m_1\times n_1}$,..., $X_d\in F^{m_d\times n_d}$.

Now, suppose $A$ is a complex cubical hypermatrix of order $N$ with side length $d$, i.e. $A\in \mathbb{C}^{\overbrace{d\times...\times d}^{N\text{ times}}}$. Then Cayley's First Hyperdeterminant \cite{lim2013tensors,amanov2023tensor}, also known as the combinatorial hyperdeterminant of $A$, is given by the function
\begin{equation}\label{hyperdeterminant - general form}
    \frac{1}{d!}\sum\limits_{\sigma_1,\sigma_2,...,\sigma_N\in S_d}\mathrm{sgn}(\sigma_1...\sigma_d)\prod\limits_{j=1}^dA_{\sigma_1(j)\sigma_2(j)...\sigma_N(j)}.
\end{equation}
Note that Cayley's first hyperdeterminant is invariant under the $\pi$-transpose. Moreover, note that it is identically $0$ for all odd order hypermatrices. For hypermatrices of even order, it is equal to 
\begin{equation}\label{hdet - even order}
    \sum\limits_{\sigma_2,...,\sigma_N\in S_d}\mathrm{sgn}(\sigma_2...\sigma_d)\prod\limits_{j=1}^dA_{j\sigma_2(j)...\sigma_N(j)} =: \mathrm{hdet}(A). 
\end{equation}
Cayley's first hyperdeterminant is also multiplicative-like under a few different hypermatrix products. Namely, 
\begin{proposition}\label{hdet - mult lin product}\cite{amanov2023tensor}
	Let $A$ be a cubical hypermatrix of even order $N$ with side length $d$, and suppose $X_1,...,X_N$ are $d\times d$ matrices. Then
    \[\mathrm{hdet}((X_1,...,X_N)*A) = \det(X_1)...\det(X_N)\mathrm{hdet}(A).\]
\end{proposition}
\begin{proposition}\label{hdet - outer product}\cite{amanov2023tensor}
    Let $A$ and $B$ be cubical hypermatrices of even orders $N-k$ and $k$ (respectively), each with side length $d$. Then 
    \[\mathrm{hdet}(A\circ B) = d!\mathrm{hdet}(A)\mathrm{hdet}(B).\]
\end{proposition}
The following proposition will also be needed:
\begin{proposition}\label{hdet vanishes when A has a vector factor}
    Let $A$ be a cubical hypermatrix of even order $N$ with side length $d$, and suppose 
    \[A = v\circ B\]
    for some $d$-dimensional vector $v$ and order $N-1$ cubical hypermatrix $B$ with side length $d$. Then, 
    \[\mathrm{hdet}(A) = 0.\]
\end{proposition}
\begin{proof}
    Under the assumption, 
    \[\mathrm{hdet}(A) = \sum\limits_{\sigma_2,...,\sigma_N\in S_d}\mathrm{sgn}(\sigma_2...\sigma_N)\prod\limits_{j=1}^dv_jB_{\sigma_2(j)...\sigma_N(j)} = \underbrace{\left(\sum\limits_{\sigma_2,...,\sigma_N\in S_d}\mathrm{sgn}(\sigma_2,...,\sigma_N)\prod\limits_{j=1}^dB_{\sigma_2(j)...\sigma_N(j)}\right)}_{d!\mathrm{hdet}(B)}\prod\limits_{j=1}^dv_j = 0,\]
    where we have used the fact that $\mathrm{hdet}$ vanishes on odd order hypermatrices. 
\end{proof}
For completeness, we note the following corollary:
\begin{corollary}
    Let $A$ be a cubical hypermatrix of even order $N$ with side length $d$, and suppose that 
    \[A = B\circ C\]
    with $B$ and $C$ odd order hypermatrices. Then 
    \[\mathrm{hdet}(A)=0.\]
\end{corollary}

\subsection{Qudit States as Hypermatrices}
An arbitrary $n$-qudit pure quantum state $\ket{\psi}$ can be expressed as the sum
\[\ket{\psi} = \sum\limits_{i_1,...,i_n=0}^{d-1}\psi_{i_1...i_n}\ket{i_1...i_n}\]
such that 
\[\sum\limits_{i_1,...,i_n=0}^{d-1}|\psi_{i_1...i_n}|^2=1.\]
Recall that $\ket{i_1...i_n} = \ket{i_1}\otimes...\otimes\ket{i_n}$. Therefore, any such state $\ket{\psi}$ may be viewed as a hypermatrix by replacing the Kronecker product $\otimes$ with the outer product $\circ$, in which case we obtain the hypermatrix 
\[\sum\limits_{i_1,...,i_n=0}^{d-1}\psi_{i_1...i_n}\ket{i_1}\circ...\circ \ket{i_n} =: \widehat{\psi}.\]
The correspondence $H:\ket{\psi}\mapsto \widehat{\psi}$ is then an injection between pure states of $n$-quits and complex cubical hypermatrices of order $n$ with side length $d$; moreover, this mapping preserves the Frobenius norm, hence it may be considered a bijection when restricting to hypermatrices with Frobenius norm $1$. Additionally, multilinear matrix multiplication is preserved under this mapping, as the proposition below shows:
\begin{proposition}\label{Kronecker and multi lin mat product}
    Let $\ket{\psi}$ and $\ket{\varphi}$ be pure states of $n$-qudits with corresponding hypermatrix representation $\widehat{\psi}$ and $\widehat{\varphi}$ (respectively). Then for any $d\times d$ matrices $M_1,...,M_n$, we have that  
    \[\ket{\varphi} = (M_1\otimes...\otimes M_n)\ket{\psi} \iff \widehat{\varphi} = (M_1,...,M_n)*\widehat{\psi}.\]
\end{proposition}
\begin{proof} 
    Observe 
    \begin{align*}
        \ket{\varphi} &= (M_1\otimes...\otimes M_n) \ket{\psi} \\
        &= (M_1\otimes...\otimes M_n)\sum\limits_{i_1,...,i_n=0}^{d-1}\psi_{i_1...i_n}\ket{i_1}\otimes...\otimes \ket{i_n} \\
        &= \sum\limits_{i_1,...,i_n=0}^{d-1}\psi_{i_1...i_n}(M_1\ket{i_1})\otimes...\otimes (M_n\ket{i_n}),\quad \text{by linearity of Kronecker products.}
    \end{align*}
    The bijection induced by $H$ maps this to
    \[\sum\limits_{i_1,...,i_n=0}^{d-1}\psi_{i_1...i_n}(M_1\ket{i_1})\circ...\circ (M_n\ket{i_n})\]
    and by the linearity of multilinear matrix multiplication, this is equal to
    \[(M_1,...,M_n)*\left(\sum\limits_{i_1,...,i_n=0}^{d-1}\psi_{i_1...i_n}\ket{i_1}\circ...\circ \ket{i_n}\right)\]
    which is precisely
    \[\widehat{\varphi}=(M_1,...,M_n)*\widehat{\psi}.\]
\end{proof}

\section{$|\mathrm{hdet}|^{2/d}$ Satisfies the Entanglement Measure Axioms}
\subsection{The Pure State Case}
\begin{proposition}\label{hdet on separable pure states}
    If $\ket{\psi}$ is a separable pure state of a $2n$-qudit, then $\mathrm{hdet}(\widehat{\psi}) = 0$.
\end{proposition}
\begin{proof}
    If $\ket{\psi}$ is separable, then $\ket{\psi}$ may be factored as $\ket{\psi} = \ket{\phi}\otimes \ket{\varphi}$ with $\ket{\phi}$ a pure qudit state and $\ket{\varphi}$ a pure state of a $(2n-1)$-qudit. Then in particular, 
    \[\ket{\psi} = \left(\sum\limits_{i=0}^{d-1}\phi_i\ket{i}\right)\otimes \left(\sum\limits_{j_1,...,j_{2n-1}=0}^{d-1}\varphi_{j_1...j_{2n-1}}\ket{j_1...j_{2n-1}}\right) = \sum\limits_{i,j_1,...,j_{2n-1}=0}^{d-1}\phi_i\varphi_{j_1...j_{2n-1}}\ket{ij_1...j_{2n-1}},\]
    hence by bilinearity of the outer product it follows that 
    \[\widehat{\psi} = \ket{\phi}\circ \widehat{\varphi},\]
    in which case by Proposition \ref{hdet vanishes when A has a vector factor} we have that $\mathrm{hdet}(\widehat{\psi})=0$. 
\end{proof}
Therefore, if $\ket{\psi}$ is a separable pure state of a $2n$-qudit, then $|\mathrm{hdet}(\widehat{\psi})|^{2/d}=0$, proving that $|\mathrm{hdet}|^{2/d}$ satisfies the first entanglement measure axiom. Note, however, that by Proposition \ref{hdet vanishes when A has a vector factor}, the converse of the above proposition does not hold; consequently, $|\mathrm{hdet}|^{2/d}$ may vanish on states with non-trivial entanglement. This is (as expected) consistent with the concurrence/$n$-tangle, because in particular if $\ket{\psi}$ is any pure $2n$-qubit state which can be factored as $\ket{\psi} = \ket{\phi}\otimes \ket{\varphi}$ with $\ket{\phi}$ a pure qubit state, then 
\[C(\ket{\psi}) = |\braket{\widetilde{\phi\otimes \varphi}}{\phi\otimes \varphi}| = |\mel{\phi^*\otimes \varphi^*}{\sigma_y^{\otimes 2n}}{\phi\otimes \varphi}| = \underbrace{|\mel{\phi^*}{\sigma_y}{\phi}|}_{=0}\cdot |\mel{\varphi^*}{\sigma_y^{\otimes 2n-1}}{\varphi}| = 0.\]
Thus, the concurrence/$n$-tangle also vanishes on states which may have non-trivial entanglement (indeed, the exact same type of biseparable states which $\mathrm{hdet}$ vanishes on). 

Next, note that if $U_1,...,U_{2n} \in U(d)$, then by Propositions \ref{hdet - mult lin product} and \ref{Kronecker and multi lin mat product} we have that
\[\left|\mathrm{hdet}\big((U_1,...,U_{2n})*\widehat{\psi}\big)\right|^{2/d} = \left|\det(U_1)...\det(U_{2n})\mathrm{hdet}(\widehat{\psi})\right|^{2/d} = |\mathrm{hdet}(\psi)|^{2/d}\]
Thus, $|\mathrm{hdet}|^{2/d}$ satisfies the second entanglement measure axiom. Indeed, the theorem below shows that $|\mathrm{hdet}|^{2/d}$ satisfies the third entanglement measure axiom as well. 
\begin{theorem}\label{hdet is an entanglement monotone on 2n-qubits}
    $|\mathrm{hdet}|^{2/d}$ is an entanglement monotone on pure states of $2n$-qudits. 
\end{theorem}
\begin{proof}
    We've already shown that $|\mathrm{hdet}|^{2/d}$ is invariant under local unitary operations. Therefore the only thing left to show is that $|\mathrm{hdet}|^{2/d}$ does not increase on average under LOCC. 
    
    Any LOCC protocol can be decomposed into POVMs acting independently on the tensor factors of $(\mathbb{C}^d)^{\otimes 2n}$ \cite{dur2000three}, and since $\mathrm{hdet}$ is invariant under the $\pi$-transpose -hence invariant under permutations of the tensor factors in $(\mathbb{C}^d)^{\otimes 2n}$, it is enough to consider a POVM acting on just the first tensor factor $\mathbb{C}^d$. Moreover, we only need to consider a two-outcome  POVM since in general an $N$-outcome POVM can be implemented by a sequence of two outcome POVMs \cite{dur2000three}. Therefore consider an arbitrary two-outcome POVM $\{M_1,M_2\}$ acting on $\mathbb{C}^d$, which recall satisfies the completeness property 
    \[M_1^{\dagger}M_1+M_2^{\dagger}M_2 = I_d.\]
    Recall also that for any matrix $M$, $M^{\dagger}M$ is positive semi-definite. Consequently, if $v$ is an eigenvector of $M_1^{\dagger}M_1$ with corresponding eigenvalue $\lambda\geq 0$, then observe that 
    \begin{align}
        (M_1^{\dagger}M_1+M_2^{\dagger}M_2)v &= I_dv \label{eigenvector1} \\
        \iff \lambda v + M_2^{\dagger}M_2v &= v \label{eigenvector2}
    \end{align}
    which implies that 
    \[M_2^{\dagger}M_2v = (1-\lambda)v.\] 
    That is, $v$ is also an eigenvector of $M_2^{\dagger}M_2$ and corresponds to the eigenvalue $1-\lambda\geq 0$, which forces $\lambda \leq 1$. Moreover, it follows that both $M_1$ and $M_2$ have the same right singular vectors, hence there exists singular value decompositions of $M_1$ and $M_2$ such that 
    \begin{equation}\label{SVD1}
        M_1 = U_1\Sigma_1 V^{\dagger}
    \end{equation}
    and
    \begin{equation}\label{SVD2}
        M_2 = U_2\Sigma_2 V^{\dagger}
    \end{equation}
    for some unitary matrices $U_1,U_2,V \in U(d)$. Furthermore, if $\Sigma_1 = \mathrm{diag}(\sigma_1,...,\sigma_d)$, then from equations \eqref{eigenvector1} and \eqref{eigenvector2} it follows that $\Sigma_2 = \mathrm{diag}(\sqrt{1-\sigma_1^2},...,\sqrt{1-\sigma_d^2})$ since the right singular vectors of $M_1$ correspond to the eigenvalues $\sigma_1^2,...,\sigma_d^2$ of $M_1^{\dagger}M_1$. 

    Now, consider an arbitrary $2n$-qudit state $\ket{\psi}$ and its corresponding hypermatrix $\widehat{\psi}$. Then if for $i=1,2$, we define 
    \[\ket{\varphi_i} := (M_i\otimes I_d^{\otimes 2n-1})\ket{\psi},\]
    then observe that 
    \[\|\varphi_i\|^2 = \braket{\varphi_i}{\varphi_i} = \mel{\psi}{(M_i\otimes I_d^{\otimes 2n-1})^{\dagger}(M_i\otimes I_d^{\otimes 2n-1})}{\psi} = \mel{\psi}{M_i^{\dagger}M_i\otimes I_d^{\otimes 2n-1}}{\psi},\]
    which is precisely the probability of the $i^{th}$ POVM outcome. For each $i$, denote this quantity as $p_i$, and also define 
    \[\ket{\phi_i} := \frac{1}{\sqrt{p_i}}\ket{\varphi_i}\]
    so that it is has norm $1$ (and is hence a valid $2n$-qudit state). Then $\ket{\phi_i}$ has corresponding hypermatrix
    \[\widehat{\phi_i} = \frac{1}{\sqrt{p_i}}(M_i,I_d,...,I_d)*\widehat{\psi_i},\]
    and in particular the expected value of $|\mathrm{hdet}|^{2/d}$ after the POVM $\{M_1,M_2\}$ is given by
    \begin{equation}\label{expected value}
        \expval{|\mathrm{hdet}|^{2/d}} = p_1|\mathrm{hdet}(\widehat{\phi}_1)|^{2/d}+p_2|\mathrm{hdet}(\widehat{\phi}_2)|^{2/d}.
    \end{equation}
    Observe that
    \begin{align*}
        |\mathrm{hdet}(\widehat{\phi}_i)|^{2/d} &= \left|\mathrm{hdet}\left(\frac{1}{\sqrt{p_i}}(M_i,I_d,...,I_d)*\widehat{\psi}\right)\right|^{2/d} \\
        &= \left|\frac{1}{p_i^{d/2}}\mathrm{hdet}\left((M_i,I_d,...,I_d)*\widehat{\psi}\right)\right|^{2/d},\quad \text{since }|\mathrm{hdet}|\text{ is homogeneous of degree d} \\ 
        &= \frac{1}{p_i}\left|\mathrm{hdet}\big((M_i,I_d,...,I_d)*\widehat{\psi}\big)\right|^{d/2} \\
        &=  \frac{1}{p_i}\left|\mathrm{hdet}\left((U_i\Sigma_iV^{\dagger},I_d,...,I_d)*\widehat{\psi}\right)\right|^{2/d},\quad \text{by equations }\eqref{SVD1},\eqref{SVD2} \\
        &= \frac{1}{p_i}\left|\mathrm{hdet}\Big((U_i,I_d,...,I_d)*\big((\Sigma_iV^{\dagger},I_d,...,I_d)*\widehat{\psi}\big)\Big)\right|^{2/d} \\
        &= \frac{1}{p_i}\left|\mathrm{hdet}\big((\Sigma_iV^{\dagger},I_d,...,I_d)*\widehat{\psi}\big)\right|^{2/d},\quad \text{by Proposition 1} \\
        &= \frac{1}{p_i}\left|\mathrm{hdet}\Big((\Sigma_i,I_d,...,I_d)*\big((V^{\dagger},I_d,...,I_d)*\widehat{\psi}\big)\Big)\right|^{2/d} \\
        &= \frac{1}{p_i}\det(\Sigma_i)^{2/d}\left|\mathrm{hdet}\big((V^{\dagger},I_d,...,I_d)*\widehat{\psi}\big)\right|^{2/d},\quad \text{by Proposition 1 and fact that }\det(\Sigma_i)\geq 0 \\
        &= \frac{\det(\Sigma_i)^{2/d}}{p_i}|\mathrm{hdet}(\widehat{\psi})|^{2/d},\quad \text{by Proposition 1 once more.} 
    \end{align*}
    Thus,
    \begin{equation}\label{phi_1/phi_2}
        |\mathrm{hdet}(\widehat{\phi}_1)|^{2/d} = \frac{\left(\prod\limits_{k=0}^{d-1}\sigma_k\right)^{2/d}}{p_1}|\mathrm{hdet}(\widehat{\psi})|^{2/d}\quad\text{and}\quad |\mathrm{hdet}(\widehat{\phi}_2)|^{2/d} = \frac{\left(\prod\limits_{k=0}^{d-1}\sqrt{1-\sigma_k^2}\right)^{2/d}}{p_2}|\mathrm{hdet}(\widehat{\psi})|^{2/d}.
    \end{equation}
    Therefore (assuming $|\mathrm{hdet}(\widehat{\psi})|^{2/d}\neq 0$), by equation \eqref{expected value}, we have that 
    \begin{align*}
        \frac{\expval{|\mathrm{hdet}|^{2/d}}}{|\mathrm{hdet}(\widehat{\psi})|^{2/d}} &= \frac{p_1|\mathrm{hdet}(\widehat{\phi}_1)|^{2/d} + p_2|\mathrm{hdet}(\widehat{\phi}_2)|^{2/d}}{|\mathrm{hdet}(\widehat{\psi})|^{2/d}},\quad \text{by equation }\eqref{expected value} \\
        &= \frac{\left(\prod\limits_{k=0}^{d-1}\sigma_k\right)^{2/d}|\mathrm{hdet}(\widehat{\psi})|^{2/d}+\left(\prod\limits_{k=0}^{d-1}\sqrt{1-\sigma_k^2}\right)^{2/d}|\mathrm{hdet}(\widehat{\psi})|^{2/d}}{|\mathrm{hdet}(\widehat{\psi})|^{2/d}},\quad \text{by line }\eqref{phi_1/phi_2} \\
        &= \left(\prod\limits_{k=0}^{d-1}\sigma_k^2\right)^{1/d} + \left(\prod\limits_{k=0}^{d-1}(1-\sigma_k^2)\right)^{1/d} \\
        &\leq \frac{1}{d}\sum\limits_{k=0}^{d-1}\sigma_k^2+\frac{1}{d}\sum\limits_{k=0}^{d-1}(1-\sigma_k^2),\quad \text{by the arithmetic-geometric mean inequality} \\
        &= 1,
    \end{align*}
    which proves that $|\mathrm{hdet}|^{2/d}$ does not increase on average under LOCC. Thus, $|\mathrm{hdet}|^{2/d}$ is an entanglement monotone. 
\end{proof}
Thus, not only does $2|\mathrm{hdet}|$ reduce to the concurrence on $2n$-qubit states, Cayley's first hyperdeterminant in fact generalizes the concurrence to $2n$-qudit states because $|\mathrm{hdet}|^{2/d}$ satisfies the entanglement measure axioms, and is should thus be regarded as an entanglement measure in its own right. Below we consider an example to elucidate the particular type of entanglement Cayley's first hyperdeterminant is detecting. 

\begin{eg}
    Let $p\in [0,1]$ and consider the $2$-qubit pure quantum state
    \[\ket{\psi} = \sqrt{p}\ket{00}+\sqrt{1-p}\ket{11}.\]
    The concurrence of $\ket{\psi}$ is given by $2\sqrt{p(1-p)}$, which is minimized when $p=0$ or $p=1$, in which case $\ket{\psi}$ is separable, and is maximized when $p=\frac{1}{2}$, in which case $\ket{\psi}$ reduces to the maximally entangled Bell state $\ket{\Phi^+}$. The $2n$-qubit extension of $\ket{\psi}$ is given by 
    \[\ket{\psi} = \sqrt{p}\ket{0}^{\otimes 2n}+\sqrt{1-p}\ket{1}^{\otimes 2n},\]
    which likewise has concurrence $2\sqrt{p(1-p)}$ and is maximized when $p=\frac{1}{2}$, yielding the generalized GHZ $2n$-qubit pure state
    \[\ket{GHZ_{2n,2}} = \frac{1}{\sqrt{2}}(\ket{0}^{\otimes 2n}+\ket{1}^{\otimes 2n}).\]
    Generalizing further, we may consider the pure $2n$-qudit state
    \[\ket{\psi} = \sum\limits_{i=0}^{d-1}\sqrt{\lambda_i}\ket{i}^{\otimes 2n}\]
    with $\sum\limits_{i=0}^{d-1}\lambda_i=1$. The corresponding hypermatrix $\widehat{\psi}$ is diagonal, hence 
    \[|\mathrm{hdet}(\widehat{\psi})|^{2/d} = \left(\prod\limits_{i=0}^{d-1}\lambda_i\right)^{1/d},\]
    which is maximized when $\lambda_i = \frac{1}{d}$ for each $i=0,1,...,d-1$. In such case, we obtain the generalized GHZ $2n$-qudit state 
    \[\ket{GHZ_{2n,d}} = \frac{1}{\sqrt{d}}\sum\limits_{i=0}^{d-1}\ket{i}^{\otimes 2n}.\]
    Normalizing, we define
    \[\mu(\ket{\psi}) := d|\mathrm{hdet}(\widehat{\psi})|^{2/d},\]
    in which case we obtain $\mu(\ket{GHZ_{2n,d}})=1$, analogous to the fact that $C(\ket{GHZ_{2n,2}})=1$ and $C(\ket{\Phi^+})=1$. Moreover, note that when $n=1$, $\mu$ is precisely the $G$-concurrence first defined in \cite{gour2005mixed}.  
\end{eg}
Thus, Cayley's first hyperdeterminant is not merely a higher-dimensional extension of the concurrence from $2n$-qubits to $2n$-qudits, it is also a multipartite extension of the $G$-concurrence from $2$-qudits to $2n$-qudits. In particular, the normalized measure $\mu = d|\mathrm{hdet}|^{2/d}$ measures the amount of full $d$-level, $2n$-partite GHZ-type entanglement that is present in a state. Furthermore, like the $G$-concurrence, it is sensitive to whether or not all $d$ levels are involved in the entanglement. Indeed, in the next example we provide a specific instance of it vanishing on a genuinely entangled quantum state due to it not utilizing all $d$-levels in its entanglement, analogous to the $G$-concurrence vanishing on any state with Schmidt rank less than $d$ \cite{sentis2016quantifying}. 

In addition to detecting genuine full $d$-level $2n$-partite GHZ-type entanglement, $\mu$ is also helpful in establishing upper bounds for probabilities in converting a given $2n$-qudit state into a generalized GHZ-like $2n$-qudit state by LOCC. In general, if a LOCC protocol transforms a pure quantum state $\ket{\psi}$ into a mixture of the states $\ket{\phi}$, $\ket{\varphi_1},\dots,\ket{\varphi_m}$ with corresponding probabilities $p,q_1,\dots,q_m$, then $\mu$ being non-increasing on average under LOCC means that 
\[p\mu(\ket{\phi})+\sum\limits_{i=1}^mq_i\mu(\ket{\varphi_i}) \leq \mu(\ket{\psi}),\]
in which case it must follow that 
\[p\mu(\ket{\phi})\leq \mu(\ket{\psi})\]
since entanglement monotones are nonnegative. Therefore, suppose again that 
\[\ket{\psi} = \sum\limits_{i=0}^{d-1}\sqrt{\lambda_i}\ket{i}^{\otimes 2n}\]
with $\sum\limits_{i=0}^{d-1}\lambda_i$. 
From the discussion and example above, 
a stochastic LOCC protocol transforms $\ket{\psi}$ into $\ket{GHZ_{2n,d}}$ with probability $p$ such that 
\[p\mu(\ket{GHZ}_{2n,d}) \leq \mu(\ket{\psi}),\]
or equivalently such that 
\begin{equation}\label{GHZ bound}
    p \leq d\left(\prod\limits_{i=0}^{d-1}\lambda_i\right)^{1/d}
\end{equation}
\begin{eg}
    Let $\epsilon\in [0,1]$ and consider the $4$-qutrit pure quantum state
    \[\ket{\psi_{\epsilon}} = \sqrt{\frac{1-\epsilon}{2}}\ket{0000}+\sqrt{\frac{1-\epsilon}{2}}\ket{1111}+\sqrt{\epsilon}\ket{2222}.\]
    When $\epsilon=0.1$, by equation \eqref{GHZ bound} any stochastic LOCC protocol converting $\ket{\psi_{\epsilon}}$ into the generalized GHZ $4$-qutrit $\ket{GHZ_{4,3}}$ can succeed with probability at most 
    \[p = 3\sqrt[3]{(0.45)^2(0.1)}\approx 0.818.\] 
    When $\epsilon=0$, $\ket{\psi_{\epsilon}}$ reduces to the embedding of the GHZ $4$-qubit $\ket{GHZ_{4,2}}$ into $(\mathbb{C}^3)^{\otimes 4}$ and vanishes on $\mu$; hence, by equation \eqref{GHZ bound} no stochastic LOCC protocol can transform this state into $\ket{GHZ_{4,3}}$ with nonzero probability. This is consistent with the fact that $\ket{\psi_0}$ contains only two-level GHZ entanglement, whereas $\ket{GHZ_{4,3}}$ contains three-level GHZ entanglement, and in the multipartite setting a LOCC protocol cannot generate genuinely higher-dimensional multipartite entanglement from lower-dimensional multipartite entanglement alone \cite{huber2013structure} (analogous to the fact that in the bipartite case, LOCC cannot generate entanglement when there is none). 
\end{eg}

\subsection{Extending to Mixed States}
We can extend Cayley's first hyperdeterminant to mixed states of $2n$-qudits by considering its convex roof extension. In general, if $E$ is any entanglement monotone on pure quantum states, then its \textbf{convex roof extension} is defined as 
\[E(\rho) := \min\left\{\sum\limits_ip_iE(\ket{\psi_i}):\rho = \sum\limits_ip_i\ket{\psi_i}\bra{\psi_i}\right\},\]
and it is a convex entanglement monotone of mixed states \cite{eltschka2014quantifying}. Therefore, if $\rho$ is a mixed state of $2n$-qudits, then the convex roof extension of the normalized measure $\mu$ on mixed states is defined as 
\[\mu(\rho) := \min\left\{\sum\limits_ip_i\mu(\ket{\psi_i}):\rho = \sum\limits_ip_i\ket{\psi_i}\bra{\psi_i}\text{ and }\mu(\ket{\psi}) = d|\mathrm{hdet}(\widehat{\psi})|^{2/d}\right\}.\]
Since $\mu$ was shown to be entanglement monotones on pure states of $2n$-qudits, by the convex roof extension $\mu$ is also automatically a convex entanglement monotone on mixed states of $2n$-qudits. Furthermore, recall that a separable mixed state $\rho$ can be expressed as 
\[\rho = \sum\limits_ip_i\ket{\psi_i^1}\bra{\psi_i^1}\otimes...\otimes\ket{\psi_i^{2n}}\bra{\psi_i^{2n}} = \sum\limits_ip_i\left(\ket{\psi_i^1}\otimes...\otimes \ket{\psi_i^{2n}}\right)\left(\bra{\psi_i^1}\otimes...\otimes\bra{\psi_i^{2n}}\right),\]
where the right side follows from repeated applications of the matrix property $(AC)\otimes (BD) = (A\otimes B)(C\otimes D)$. Therefore if $\rho$ is separable, then 
\[0 \leq \mu(\rho) \leq \sum\limits_ip_i\underbrace{\mu\left(\ket{\psi_i^1}\circ...\circ \ket{\psi_i^{2n}}\right)}_{=\text{ }0\text{ by Proposition \ref{hdet on separable pure states}}} = 0.\]
Thus, we have that $\mu$ vanishes on separable mixed states as well, and so it therefore may be considered to be an entanglement measure on mixed states of $2n$-qudits as well. 

\begin{eg}
    Let $p\in [0,1]$ and consider the mixed state 
    \[\rho_p = p\ket{GHZ_{2n,3}}\bra{GHZ_{2n,3}}+(1-p)\sigma,\]
    where
    \[\ket{GHZ_{2n,3}} := \frac{1}{\sqrt{3}}\left(\ket{0}^{\otimes 2n}+\ket{1}^{\otimes 2n}+\ket{2}^{\otimes 2n}\right)\]
    and
    \[\sigma := \frac{1}{3}\left((\ket{0}\bra{0})^{\otimes 2n}+(\ket{1}\bra{1})^{\otimes 2n}+(\ket{2}\bra{2})^{\otimes 2n}\right).\]
    Note that $\rho_p$ is a linear operator on the subspace $S := \mathrm{span}\{\ket{0}^{\otimes 2n},\ket{1}^{\otimes 2n},\ket{2}^{\otimes 2n}\}$ with matrix representation
    \[\frac{1}{3}\left[\begin{array}{ccc}
        1 & p & p \\
        p & 1 & p \\
        p & p & 1
    \end{array}\right]\]
    Additionally, when $p=1$, $\rho_1 = \ket{GHZ_{2n,3}}\bra{GHZ_{2n,3}}$ represents the pure $2n$-qutrit GHZ state $\ket{GHZ_{2n,3}}$ which is maximally entangled with respect to $\mu$, hence $\mu(\rho_1)=1$. When $p=0$, $\rho_0 = \sigma$, which is separable, hence $\mu(\rho_0)=0$. More interestingly, $\mu$ remains zero until $p > \frac{1}{2}$. 
    
    Indeed, when $p=\frac{1}{2}$, 
    \[\rho_{1/2} = \frac{1}{3}
    \left(\ket{GHZ_{2n,2}^1}\bra{GHZ_{2n,2}^1}
    + \ket{GHZ_{2n,2}^2}\bra{GHZ_{2n,2}^2}
    + \ket{GHZ_{2n,2}^3}\bra{GHZ_{2n,2}^3}
    \right),\]
    where
    \[\begin{cases}
        \begin{rcases}
            \ket{GHZ_{2n,2}^1} := \frac{1}{\sqrt{2}}\left(\ket{0}^{\otimes 2n}+\ket{1}^{\otimes 2n}\right) \\
            \ket{GHZ_{2n,2}^2} := \frac{1}{\sqrt{2}}\left(\ket{0}^{\otimes 2n}+\ket{2}^{\otimes 2n}\right) \\
            \ket{GHZ_{2n,2}^3} := \frac{1}{\sqrt{2}}\left(\ket{1}^{\otimes 2n}+\ket{2}^{\otimes 2n}\right)
        \end{rcases},
    \end{cases}\]
    and when $p\in \left[0,\frac{1}{2}\right]$, we may write 
    \[\rho_p = 2p\rho_{1/2}+(1-2p)\sigma,\]
    and when $p\in \left(\frac{1}{2},1\right]$, we may write 
    \[\rho_p = (2p-1)\ket{GHZ_{2n,3}}\bra{GHZ_{2n,3}}+2(1-p)\rho_{1/2}.\]
    Each of the components $\ket{GHZ_{2n,2}^i}$ in $\rho_{1/2}$ are $2$-level GHZ $2n$-qutrits, and in particular their corresponding hypermatrix representation is diagonal with a zero entry, forcing $\mu$ to vanish on each of these states; hence, $\mu(\rho_{1/2})=0$. Since $\mu$ is a convex entanglement monotone on mixed states, for $p\in \left[0,\frac{1}{2}\right]$ we have that 
    \[\mu(\rho_p)\leq 2p\mu(\rho_{1/2})+(1-2p)\mu(\sigma) = 0.\]
    On the other hand, for $p\in \left(\frac{1}{2},1\right]$ we have that 
    \[\mu(\rho_p) \leq (2p-1)\mu(\ket{GHZ_{2n,3}}\bra{GHZ_{2n,3}})+2(1-p)\mu(\rho_{1/2}) = 2p-1,\]
    since $\ket{GHZ_{2n,3}}\bra{GHZ_{2n,3}}$ represents the pure state $\ket{GHZ_{2n,3}}$ and $\mu(\ket{GHZ_{2n,3}})=1$; in fact, we claim \[\mu(\rho_p)=2p-1\] 
    whenever $p\in \left(\frac{1}{2},1\right]$. To this end, we first claim that for every $2n$-qutrit pure state of the form 
    \[\ket{\psi} = \alpha\ket{0}^{\otimes 2n}+\beta\ket{1}^{\otimes 2n}+\gamma\ket{2}^{\otimes 2n},\]
    we have that 
    \[\mu(\ket{\psi})\geq 3\left|\braket{GHZ_{2n,3}}{\psi}\right|^2-2.\]
    Indeed, by definition/direct calculation
    \[\mu(\ket{\psi})=3|\alpha\beta\gamma|^{2/3}\]
    and
    \[3\left|\braket{GHZ_{2n,3}}{\psi}\right|^2-2=|\alpha+\beta+\gamma|^2-2.\]
    By the triangle inequality, $|\alpha+\beta+\gamma|\leq |\alpha|+|\beta|+|\gamma|$, hence 
    \begin{equation}\label{triangle inequality}
        |\alpha+\beta+\gamma|^2\leq (|\alpha|+|\beta|+|\gamma|)^2 = |\alpha|^2+|\beta|^2+|\gamma|^2+2(|\alpha\beta|+|\alpha\gamma|+|\beta\gamma|) = 1+2(|\alpha\beta|+|\alpha\gamma|+|\beta\gamma|),
    \end{equation}
    where the last equality is due to the fact that as a $2n$-qutrit pure state $\ket{\psi}$ must satisfy $|\alpha|^2+|\beta|^2+|\gamma|^2=1$. Now, let $a = |\alpha|^{2/3}$, $b = |\beta|^{2/3}$, and $c = |\gamma|^{2/3}$. By Schur's inequality and the fact that $x^2y+xy^2-2(xy)^{3/2} = xy(\sqrt{x}-\sqrt{y})^2\geq 0$ for all $x,y\geq 0$, we have that 
    \begin{align*}
        a^3+b^3+c^3+3abc &\geq a^2b+b^2a+a^2c+c^2a+b^2c+c^2b \\
        &\geq 2(ab)^{3/2}+2(ac)^{3/2}+2(bc)^{3/2},
    \end{align*}
    which in this case implies that  
    \[|\alpha|^2+|\beta|^2+|\gamma|^2 + 3|\alpha\beta\gamma|^{2/3}\geq 2(|\alpha\beta|+|\alpha\gamma|+|\beta\gamma|).\]
    Combining this with line \eqref{triangle inequality}, we have that 
    \begin{align*}
        3|\alpha\beta\gamma|^{2/3} &\geq 2(|\alpha\beta|+|\alpha\gamma|+|\beta\gamma|)-1 \\
        &\geq |\alpha+\beta+\gamma|^2-2,
    \end{align*}
    proving that in fact 
    \begin{equation}\label{rho_p inequality}
        \mu(\ket{\psi})\geq 3\left|\braket{GHZ_{2n,3}}{\psi}\right|^2-2.
    \end{equation}
    Now, since $\rho_p$ is a linear operator on $S$, for any decomposition 
    \[\rho_p = \sum\limits_ip_i\ket{\psi_i}\bra{\psi_i}\]
    and any $\ket{\varphi}\in S^{\perp}$, necessarily
    \[0 = \mel{\varphi}{\rho_p}{\varphi} = \sum\limits_ip_i|\braket{\varphi}{\psi}|^2,\]
    forcing $\ket{\psi_i}\in (S^{\perp})^{\perp} = S$. Consequently, each $\ket{\psi_i}$ is in superposition with the basis states $\ket{0}^{\otimes 2n}$, $\ket{1}^{\otimes 2n}$, $\ket{2}^{\otimes 2n}$, in which case the inequality in line \eqref{rho_p inequality} implies that 
    \begin{align*}
        \sum\limits_ip_i\mu(\ket{\psi_i}) &\geq \sum\limits_ip_i\left(3\left|\braket{GHZ_{2n,3}}{\psi_i}\right|^2-2\right) \\
        &= 3\sum\limits_ip_i\braket{GHZ_{2n,3}}{\psi}\braket{\psi_i}{GHZ_{2n,3}} - 2,\quad \text{since }\sum\limits_ip_i=1 \\
        &= 3\mel{GHZ_{2n,3}}{\sum\limits_ip_i\ket{\psi_i}\bra{\psi_i}}{GHZ_{2n,3}}-2 \\
        &= 3\mel{GHZ_{2n,3}}{\rho_p}{GHZ_{2n,3}}-2 \\
        &= 3\left(\frac{1+2p}{3}\right)-2 \\
        &= 2p-1.
    \end{align*}
    Thus, indeed when $p>\frac{1}{2}$, we have that 
    \[\mu(\rho_p) = 2p-1,\]
    hence in general
    \[\mu(\rho_p) = \mathrm{max}\{0,2p-1\},\]
    for all $p\in [0,1]$. 
\end{eg}
The example above shows that on mixed states $|\mathrm{hdet}|^{d/2}$, or equivalently the normalized measure $\mu$, detects full $d$-level $2n$-partite GHZ-type entanglement, provided it cannot be decomposed entirely in terms of lower-level GHZ states or unentangled states. $\mu$ therefore should be utilized when seeking to detect \textit{genuine} maximal-dimensional GHZ-type entanglement across all parties. 

\section{Conclusion}
When viewing $2n$-qubit states as hypermatrices, it was previously established that $2|\mathrm{hdet}|$ and $4|\mathrm{hdet}|^2$ reduce to the concurrence and $n$-tangle respectively, indicating that Cayley's first hyperdeterminant captures meaningful information regarding a quantum state's $2n$-way entanglement. In this paper, we rigorously proved that indeed $|\mathrm{hdet}|^{2/d}$ is an entanglement measure on $2n$-qudit states. Moreover, the examples considered clarify the physical interpretation of this entanglement measure, which is that it detects genuine full $d$-level GHZ-type entanglement across all $2n$ parties. In this sense, Cayley's first hyperdeterminant is not just a higher-dimensional generalization of the concurrence/$n$-tangle from $2n$-qubit states to $2n$-qudit states, it is also a multipartite extension of the $G$-concurrence from $2$-qudit states to $2n$-qudit states. 

We hope this paper and our results will inspire further investigation into Cayley's first hyperdeterminant and its applications to the study of quantum entanglement. One potential avenue of exploration would be to see if it is possible to derive a formula for convex roof extensions of $|\mathrm{hdet}|^{2/d}$ (or the normalized measure $\mu$) on mixed $2n$-qudit states in terms of the tensor eigenvalues, similar to the formula for $C(\rho)$ in terms of the eigenvalues of $\sqrt{\sqrt{\rho}\widetilde{\rho}\sqrt{\rho}}$ derived in Wootters widely cited paper on the entanglement of formation \cite{wootters1998entanglement}. Such a formula might greatly simplify the procedure in calculating $\mu$ on mixed states, which, as is apparent from Example 3, is currently not very easy in general. Moreover, such a formula would further validate Cayley's first hyperdeterminant as a physically significant extension of the concurrence.  
\\
\\
\noindent{\bf Data availability statement}
Any data that support the findings of this study are included within the article.

\bibliographystyle{ieeetr}
\bibliography{bibliography}

\end{document}